\documentclass[aps,pra,twocolumn,nofootinbib,superscriptaddress,floatfix]{revtex4-2}
\usepackage[T1]{fontenc}
\usepackage{lmodern,microtype}
\usepackage{amsmath,amssymb,amsfonts,amsthm,mathtools,bm}
\usepackage{booktabs,array,enumitem,graphicx}
\usepackage[hidelinks]{hyperref}
\usepackage[nameinlink,capitalize]{cleveref}
\graphicspath{{figures/}}
\hypersetup{
 pdftitle={Fixed-boost Wigner noise: strict trace-distance contraction without quantum degradability},
 pdfauthor={Maxim V. Churilov}
}

\newtheorem{theorem}{Theorem}[section]

\newtheorem{proposition}[theorem]{Proposition}
\newtheorem{corollary}[theorem]{Corollary}

\newtheorem{remark}[theorem]{Remark}

\newcommand{\CPTP}{\mathsf{CPTP}}
\newcommand{\Tr}{\operatorname{Tr}}
\newcommand{\diag}{\operatorname{diag}}
\newcommand{\conv}{\operatorname{conv}}
\newcommand{\smin}{s_{\min}}
\newcommand{\id}{\mathrm{id}}
\newcommand{\TV}{d_{\mathrm{TV}}}

\begin{document}
\title{Fixed-Boost Wigner Noise: Strict Trace-Distance Contraction without Quantum Degradability}
\author{Maxim V. Churilov}
\email{churilovm1305@gmail.com}
\affiliation{Independent Researcher, Orenburg, Russia}
\date{November 24, 2025}

\begin{abstract}
A Lorentz boost acts on the canonical spin of a massive particle through a momentum-dependent Wigner rotation.  We show that, for one fixed observer boost, reducing over an uncertain momentum can strictly contract every pairwise spin-state trace distance without producing a channel that is degradable from the less contracted one.  For spin $1/2$, we first characterize the exact inversion-symmetric channel cone generated by a fixed Wigner angle and transverse momentum directions.  Inside this cone lies the Pauli family
\begin{equation*}
 M_\alpha=\diag(1-\alpha,1-\alpha,1-2\alpha),
 \qquad 0\leq\alpha<\tfrac12.
\end{equation*}
For $0<\alpha<\beta<1/2$, all trace distances between distinct spin states are strictly smaller after $M_\beta$ than after $M_\alpha$, yet the unique linear post-processing factor has a negative normalized Choi eigenvalue.  We solve the optimization over all physical converters exactly:
\begin{equation*}
 \frac12\inf_{\Lambda\in\CPTP}
 \|\Phi_\beta-\Lambda\Phi_\alpha\|_\diamond
 =\frac{\alpha(\beta-\alpha)}{2-3\alpha},
\end{equation*}
whereas the reverse deficiency is $\beta-\alpha$.  Thus the identity dominates the family, while all positive-noise members are pairwise incomparable under CPTP post-processing.  The ideal construction is realized as the narrow-packet limit of pure, normalizable five-component momentum states, and explicit perturbation and finite-shot tomography bounds certify an open set of examples.  Separately, every nonidentity member fails embedding in a time-homogeneous Pauli-diagonal Lindblad semigroup.  Hence ordering all unassisted spin distinguishabilities does not determine the quantum statistical post-processing order.
\end{abstract}

\maketitle
\section{Introduction}

For a massive particle, the unitary representation of a Lorentz transformation $\Lambda$ acts as
\begin{equation}
 U(\Lambda)|p,\sigma\rangle
 =\sum_{\sigma'}D_{\sigma'\sigma}\!\left(W(\Lambda,p)\right)
 |\Lambda p,\sigma'\rangle,
 \label{eq:wigner-action}
\end{equation}
where $W(\Lambda,p)$ is the Wigner rotation.  If momentum is not resolved, a spin state prepared independently of momentum evolves by a random-unitary channel.  In Bloch coordinates,
\begin{equation}
 \bm r\longmapsto M\bm r,
 \qquad
 M=\int R_W(\Lambda,p)\,\mu(dp),
 \label{eq:bloch-average}
\end{equation}
with $R_W\in SO(3)$.

Common-axis models reduce \cref{eq:bloch-average} to dephasing and can be ordered by one complex coherence parameter.  A fixed boost with a genuinely angular momentum distribution is different: Wigner axes vary across the transverse plane, branch rotations do not commute, and complete positivity couples the three observed contractions.  The central question is operational.  If a target preparation yields a smaller spin contrast than a source preparation along every calibrated axis under the same boost, must there exist a state-independent channel $\Lambda$ such that
\begin{equation}
 \Phi_B=\Lambda\circ\Phi_A?
 \label{eq:degradability-question}
\end{equation}
We prove that the answer is no, already for one fixed boost and a five-branch momentum model, and that the separation persists for pure normalizable wave packets.

The distinction matters because componentwise damping is often used as a proxy for information loss.  Quantum post-processing is stronger: the same map must act consistently on the whole Bloch ball and when the spin is entangled with an arbitrary ancilla.  A factor can contract every coordinate and still fail complete positivity.

Our analysis has three parts.  First, we isolate a compatibility domain on which the reduced spin dynamics is an honest channel.  Second, we characterize the complete convex cone generated by inversion-symmetric transverse momenta at a fixed Wigner angle.  Third, inside that cone we construct an analytic family for which all contrasts decrease but the unique factor violates one Fujiwara--Algoet facet.  Unlike a mixture of different observer boosts, the construction changes only the momentum law of the same boosted observer.

\subsection{Relation to prior work and claim boundary}

Observer dependence of reduced spin states, spin--momentum entanglement, and Wigner rotations are well established \cite{Peres2002,GingrichAdami2002,Terno2004,Caban2005,Jordan2004}.  Their operational interpretation depends on the chosen spin observable and measurement model, a point emphasized in Refs.~\cite{SaldanhaVedral2012,SaldanhaVedral2013}; throughout this work we use canonical spin and the explicit product-preparation channel defined below.  The geometry of unital qubit channels and the Fujiwara--Algoet inequalities is standard \cite{FujiwaraAlgoet1999,Ruskai2002,Braun2014}.  Quantum comparison and channel divisibility identify CPTP post-processing, rather than scalar damping, as the relevant information order \cite{Shmaya2005,Buscemi2012,WolfCirac2008}.  Exact discrimination of Pauli mixtures supplies the norm identity used below \cite{Sacchi2005,Watrous2018}.  Wigner rotation also has concrete experimental proposals for massive spin-$1/2$ particles \cite{Palge2025}.

The general factor criterion below is an application of known qubit-channel geometry.  The new result is the fixed-boost realization: an exact reachable cone and an open family of physically generated channel pairs for which all principal spin contrasts decrease while degradability fails.  We do not claim a characterization of every channel produced by arbitrary spin--momentum correlations or every Lorentz wave packet.

\paragraph{Status of the principal ingredients.}
The factor test uses standard affine qubit-channel geometry.  The central new results are the exact fixed-boost cone, the one-observer contraction counterexample with analytic Choi margin, and its persistence for normalizable wave packets.  The tomography bound is a finite-error implementation of that no-go.

\section{Compatibility domain and reduced spin channels}

Let $\mathcal H_P=L^2(\mathcal H_m^+,d\nu_m)$ be the one-particle momentum space on the positive-energy mass shell with invariant measure $d\nu_m$.  We assume an autonomous product preparation
\begin{equation}
 \rho_{SP}=\rho_S\otimes\rho_P,
 \qquad
 \rho_P\ge0,\quad \Tr\rho_P=1,
 \label{eq:compatibility}
\end{equation}
where $\rho_P$ is a genuine trace-class operator.  For a pure packet,
\begin{align}
 \rho_P&=|\psi\rangle\!\langle\psi|,
 &\psi&\in L^2(\mathcal H_m^+,d\nu_m),\nonumber\\
 \int|\psi(p)|^2d\nu_m(p)&=1.
 \label{eq:pure-packet}
\end{align}
After the Lorentz transformation and momentum trace, the off-diagonal momentum kernel disappears because $p\mapsto\Lambda p$ is one-to-one.  The spin map is
\begin{equation}
 \Phi_{\Lambda,\mu}(\rho_S)
 =\int D(W(\Lambda,p))\rho_SD(W(\Lambda,p))^\dagger\,\mu(dp).
 \label{eq:reduced-channel}
\end{equation}
Here $\mu(dp)=|\psi(p)|^2d\nu_m(p)$ for \cref{eq:pure-packet}; for a mixed trace-class state it is the corresponding momentum probability measure.  Thus \cref{eq:reduced-channel} is a Bochner integral of unitary conjugations and is unital and CPTP on the entire spin state space.  Initial spin--momentum correlations lie outside this autonomous compatibility domain and are not assigned a reduced spin channel here.  This restriction also avoids attributing operational meaning to a reduced spin state independently of the specified preparation and canonical-spin measurement model.

The formal expression $\int|p\rangle\!\langle p|\,\mu(dp)$ is deliberately avoided.  For a nonatomic momentum law it represents a multiplication operator rather than a trace-class density operator.  The wave-packet formulation above produces the same diagonal probability law after momentum trace without leaving the Hilbert-space state space.

Write a qubit state as $\rho=(I+\bm r\cdot\bm\sigma)/2$.  Then \cref{eq:reduced-channel} has Bloch matrix \cref{eq:bloch-average}.  Convex averages of Wigner rotations lie in $\conv(SO(3))$, which is precisely the Bloch body of mixed-unitary qubit channels.

\section{Exact convertibility for invertible reduced spin channels}

\begin{theorem}[Unique factor criterion]
\label{thm:factor}
Let $\Phi_A,\Phi_B$ be unital qubit channels with Bloch matrices $M_A,M_B$, and assume $M_A$ is invertible.  A channel $\Lambda$ satisfies $\Phi_B=\Lambda\Phi_A$ if and only if the unital map with Bloch matrix
\begin{equation}
 T=M_BM_A^{-1}
 \label{eq:factor}
\end{equation}
is completely positive.  The factor is unique.
\end{theorem}

\begin{proof}
A general qubit channel acts affinely as $\bm x\mapsto T\bm x+\bm t$.  Composition with $\Phi_A$ gives $M_B\bm r=TM_A\bm r+\bm t$ for every $\bm r$.  The maximally mixed input forces $\bm t=0$, and invertibility gives \cref{eq:factor}.  Complete positivity is then necessary and sufficient.
\end{proof}

Every real $3\times3$ matrix has a signed singular-value decomposition
\begin{equation}
 T=R_1\diag(\eta_1,\eta_2,\eta_3)R_2,
 \qquad R_1,R_2\in SO(3).
\end{equation}
The rotations correspond to unitary conjugations.  The factor is CPTP exactly when
\begin{align}
 1+\eta_3&\ge |\eta_1+\eta_2|,\label{eq:FA1}\\
 1-\eta_3&\ge |\eta_1-\eta_2|.
 \label{eq:FA2}
\end{align}
Equivalently, the four normalized Choi eigenvalues
\begin{align}
 p_0&=\tfrac14(1+\eta_1+\eta_2+\eta_3),\nonumber\\
 p_1&=\tfrac14(1+\eta_1-\eta_2-\eta_3),\nonumber\\
 p_2&=\tfrac14(1-\eta_1+\eta_2-\eta_3),\nonumber\\
 p_3&=\tfrac14(1-\eta_1-\eta_2+\eta_3)
 \label{eq:choi-eigs}
\end{align}
are nonnegative.  Coordinatewise bounds $|\eta_i|\le1$ do not imply these coupled inequalities.

\section{The inversion-symmetric cone of one fixed boost}

Fix an observer boost of rapidity $\xi$ in the $z$ direction.  Take sharp particle momenta with a common rapidity $\zeta$ and directions in the transverse plane.  For perpendicular canonical boosts, the Wigner angle satisfies
\begin{equation}
 \tan\frac\theta2
 =\tanh\frac\xi2\tanh\frac\zeta2,
 \qquad 0<\theta<\frac\pi2,
 \label{eq:wigner-angle}
\end{equation}
and its axis is parallel to $\hat{\bm z}\times\hat{\bm p}$.  Opposite momenta give opposite axes and hence rotations $R_{\bm n}(\pm\theta)$.

Let $s=1-\cos\theta$.  The symmetric pair average is
\begin{equation}
 \frac{R_{\bm n}(\theta)+R_{\bm n}(-\theta)}2
 =\cos\theta\,I+s\,\bm n\bm n^T.
 \label{eq:pair-average}
\end{equation}
For transverse $\bm n$, the $z$ direction is an eigenvector.

\begin{theorem}[Exact fixed-boost symmetric cone]
\label{thm:cone}
A real matrix $M$ is generated by the identity branch and an inversion-symmetric probability distribution of fixed-angle Wigner rotations whose axes lie in the plane orthogonal to the fixed boost if and only if, in coordinates with boost axis $z$,
\begin{equation}
 M=
 \begin{pmatrix}
 I_2-D&0\\
 0&1-\Tr D
 \end{pmatrix},
 \qquad
 D\succeq0,
 \quad
 \Tr D\le s.
 \label{eq:fixed-cone}
\end{equation}
Every such $M$ is realized by at most two opposite-axis pairs and the identity.
\end{theorem}

\begin{proof}
Let the total nonidentity mass be $w$ and let
\begin{equation}
 Q=\int \bm n\bm n^T\,\nu(d\bm n)
\end{equation}
be its unnormalized transverse second moment, so $Q\succeq0$, $\Tr Q=w$, and $Q\preceq wI_2$.  Averaging \cref{eq:pair-average} and adding identity mass $1-w$ gives
\begin{equation}
 M_\perp=I_2-s(wI_2-Q),
 \qquad M_{zz}=1-sw.
\end{equation}
Thus $D=s(wI_2-Q)\succeq0$ and $\Tr D=sw\le s$.

Conversely, let $D\succeq0$ with $d=\Tr D\le s$.  Put $w=d/s$ and $Q=wI_2-D/s$.  If the eigenvalues of $D$ are $d_1,d_2$, then those of $Q$ are $d_2/s,d_1/s$, so $Q\succeq0$ and $\Tr Q=w$.  Its spectral decomposition writes $Q$ as the second moment of at most two transverse unit axes with total mass $w$.  Formula \cref{eq:pair-average} then reconstructs \cref{eq:fixed-cone}.
\end{proof}

The theorem is a reachability statement for one observer boost, not merely an abstract mixed-unitary decomposition.  It also shows the kinematic constraint missed by a three-axis model: for a fixed boost, the longitudinal deficit equals the trace of the transverse deficit matrix.

\begin{proposition}[Extreme channels and minimal branch mass]
\label{prop:extreme-cone}
The extreme points of the symmetric fixed-boost cone are the identity
and the full-mass opposite-momentum pair averages
\begin{equation}
 M_{\bm n}
 =\cos\theta\,I+s\,\bm n\bm n^T,
 \qquad \bm n\perp\hat{\bm z}.
\end{equation}
For a channel with deficit matrix $D$, every realization uses total
nonidentity probability
\begin{equation}
 w=\frac{\Tr D}{s}.
 \label{eq:branch-mass}
\end{equation}
Two opposite-axis pairs always suffice and are generically necessary
when $D$ has rank two.
\end{proposition}

\begin{proof}
The affine parameter body is
$\mathcal D_s=\{D\succeq0:\Tr D\le s\}$.  Its extreme points are
$0$ and $s|\bm m\rangle\!\langle\bm m|$: a nonzero matrix of trace
strictly below $s$ can be perturbed radially, while a trace-$s$ matrix
of rank two can be perturbed along its two spectral projectors.
Conversely, rank-one trace-$s$ matrices are extreme in the
positive-semidefinite trace ball.  Under the inverse construction in
\cref{thm:cone}, $D=s|\bm m\rangle\!\langle\bm m|$ corresponds to
$w=1$ and the Wigner axis orthogonal to $\bm m$ in the transverse
plane.  The longitudinal entry of every realization is $1-sw$, while
\cref{eq:fixed-cone} gives $1-\Tr D$, proving
\cref{eq:branch-mass}.  Finally, the rank of the transverse second
moment requires two axes for generic rank-two $D$, and its spectral
decomposition supplies them.
\end{proof}

\section{Strict trace-distance contraction without degradability}

Choose two orthogonal transverse Wigner axes, generated for example by momenta along $\pm x$ and $\pm y$.  Give each opposite-axis pair total weight $w$ and the identity branch weight $1-2w$.  The resulting matrix is
\begin{equation}
 M(w)=\diag(1-sw,1-sw,1-2sw),
 \qquad 0\le w\le\frac12.
 \label{eq:Mw}
\end{equation}

\begin{theorem}[Fixed-boost contraction no-go]
\label{thm:nogo}
Fix $0<\theta<\pi/2$ and $0<u<v<1/2$.  Let the source and target preparations have the same boost $\Lambda$ and the same five momentum directions, but pair weights $u$ and $v$, respectively.  Then:
\begin{enumerate}[label=(\roman*),leftmargin=2.2em]
\item both reduced spin maps are physical channels in the fixed-boost cone;
\item every equal-prior binary discrimination problem between two distinct spin states is strictly harder for the target:
\begin{equation}
 \|\Phi_v(\rho)-\Phi_v(\sigma)\|_1
 <
 \|\Phi_u(\rho)-\Phi_u(\sigma)\|_1
 \qquad(\rho\ne\sigma);
 \label{eq:all-pair-contraction}
\end{equation}
\item no CPTP map converts the source channel into the target channel;
\item the minimum normalized Choi eigenvalue of the unique factor is
\begin{equation}
 \boxed{
 \lambda_{\min}
 =-\frac{s^2u(v-u)}{2(1-su)(1-2su)}<0 }.
 \label{eq:negative-eig}
\end{equation}
\end{enumerate}
\end{theorem}

\begin{proof}
Physicality follows from the explicit probability weights and \cref{thm:cone}.  Since $v>u$ and $s>0$, all three positive diagonal entries decrease strictly.  For $\rho-\sigma=(\bm x\cdot\bm\sigma)/2$, the qubit trace norm is
\begin{equation}
 \|\Phi_w(\rho)-\Phi_w(\sigma)\|_1=|M(w)\bm x|_2.
\end{equation}
Every coefficient of $|M(v)\bm x|_2^2$ is strictly smaller than the corresponding coefficient for $u$, proving \cref{eq:all-pair-contraction} for every $\bm x\ne0$.  The source is invertible because $2su<1$.  The unique factor is
\begin{equation}
 T=\diag(t,t,z),
 \quad
 t=\frac{1-sv}{1-su},
 \quad
 z=\frac{1-2sv}{1-2su}.
 \label{eq:T}
\end{equation}
It satisfies $0<z<t<1$.  Three normalized Choi eigenvalues are positive:
\begin{equation}
 p_0=\frac{1+2t+z}{4}>0,
 \qquad
 p_1=p_2=\frac{1-z}{4}>0.
\end{equation}
The fourth is
\begin{align}
 p_3&=\frac{1-2t+z}{4}\\
 &=-\frac{s^2u(v-u)}{2(1-su)(1-2su)},
\end{align}
which proves the no-go by \cref{thm:factor}.
\end{proof}

\begin{figure}[t]
\centering
\includegraphics[width=\columnwidth]{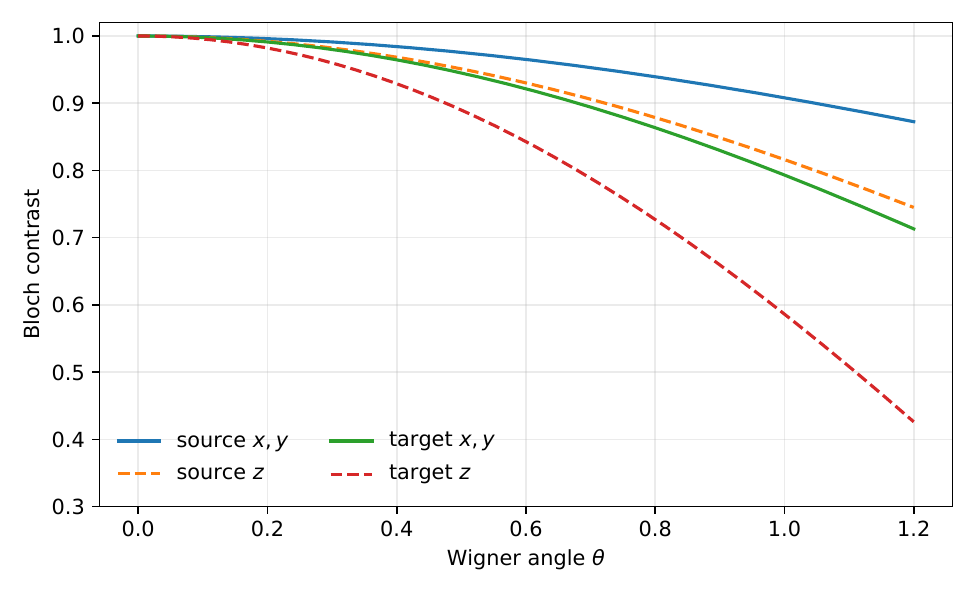}
\caption{All three target contrasts lie below the source contrasts for one fixed boost, while the unique factor has a negative Choi eigenvalue.  Parameters shown: $u=0.20$, $v=0.45$.}
\label{fig:contraction}
\end{figure}

\begin{proposition}[Strict order of standard damping summaries]
\label{prop:scalar-monotones}
On the parameter interval $0\le\alpha<1/2$, let
$M_\alpha=\diag(1-\alpha,1-\alpha,1-2\alpha)$.  Each of the following
quantities is strictly decreasing in $\alpha$:
\begin{align}
 \eta_{\rm tr}(\Phi_\alpha)
 &:=\sup_{\rho\ne\sigma}
 \frac{\|\Phi_\alpha(\rho)-\Phi_\alpha(\sigma)\|_1}
      {\|\rho-\sigma\|_1}
 =1-\alpha,\label{eq:trace-contraction-coeff}\\
 s_{\min}(M_\alpha)&=1-2\alpha,\label{eq:min-singular}\\
 \det M_\alpha&=(1-\alpha)^2(1-2\alpha),\label{eq:bloch-volume}\\
 \mathfrak u(\Phi_\alpha)
 &:=\frac13\Tr(M_\alpha^TM_\alpha)
 =\frac{2(1-\alpha)^2+(1-2\alpha)^2}{3}.
 \label{eq:unitarity}
\end{align}
Every interior member is not entanglement breaking, since
$2(1-\alpha)+(1-2\alpha)>1$.  Nevertheless, by
\cref{cor:antichain}, no two distinct positive-noise members are
CPTP-comparable.
\end{proposition}

\begin{proof}
The trace contraction coefficient of a unital qubit channel is the
largest singular value of its Bloch matrix, giving
\cref{eq:trace-contraction-coeff}.  The remaining formulas are direct.
Their derivatives are strictly negative on $(0,1/2)$.  A unital qubit
channel is entanglement breaking exactly when the sum of its singular
values is at most one; here that sum is $3-4\alpha>1$.
\end{proof}

\section{Exact physical post-processing deficiency}
\label{sec:exact-deficiency}

Nonexistence of an exact factor is a qualitative statement.  We now solve the quantitative optimization over \emph{all} physical converters.  Put
\begin{equation}
 \alpha=su,\qquad \beta=sv,
 \qquad 0\le\alpha,\beta<\frac12,
 \label{eq:alpha-beta}
\end{equation}
and write the channel as the Pauli mixture
\begin{equation}
 \Phi_\alpha(\rho)
 =(1-\alpha)\rho+\frac{\alpha}{2}X\rho X
 +\frac{\alpha}{2}Y\rho Y.
 \label{eq:pauli-line}
\end{equation}
Its Bloch matrix is exactly $\diag(1-\alpha,1-\alpha,1-2\alpha)$.

\begin{theorem}[Exact directed CPTP deficiency]
\label{thm:exact-deficiency}
For the Pauli line \cref{eq:pauli-line}, define
\begin{equation}
 \delta_\diamond(\beta|\alpha)
 =\frac12\inf_{\Lambda\in\CPTP}
 \|\Phi_\beta-\Lambda\circ\Phi_\alpha\|_\diamond.
 \label{eq:deficiency}
\end{equation}
Then
\begin{equation}
 \boxed{
 \delta_\diamond(\beta|\alpha)=
 \begin{cases}
  \alpha-\beta,&0\le\beta<\alpha,\\[2pt]
  0,&\beta=\alpha,\\[2pt]
  \displaystyle\frac{\alpha(\beta-\alpha)}{2-3\alpha},
       &\alpha<\beta<\frac12.
 \end{cases}}
 \label{eq:exact-deficiency}
\end{equation}
For $\beta>\alpha$, an optimal converter is the Pauli channel with error law
\begin{equation}
 r_*=(1-2x_*,x_*,x_*,0),
 \qquad
 x_*=\frac{\beta-\alpha}{2-3\alpha}.
 \label{eq:optimal-converter}
\end{equation}
For $\beta<\alpha$, the identity converter is optimal.
\end{theorem}

\begin{proof}
Let $P_0=I,P_1=X,P_2=Y,P_3=Z$.  Pauli twirling any converter gives
\begin{equation}
 \overline\Lambda
 =\frac14\sum_{j=0}^3
 \operatorname{Ad}_{P_j}\circ\Lambda\circ\operatorname{Ad}_{P_j}.
 \label{eq:converter-twirl}
\end{equation}
Both $\Phi_\alpha$ and $\Phi_\beta$ commute with every Pauli conjugation.  Convexity and unitary invariance of the diamond norm therefore imply
\begin{equation}
 \|\Phi_\beta-\overline\Lambda\Phi_\alpha\|_\diamond
 \le\|\Phi_\beta-\Lambda\Phi_\alpha\|_\diamond.
\end{equation}
The twirled converter is a Pauli channel, represented by a probability vector $r$ on the four-element Pauli group modulo phases.  If $p=(1-\alpha,\alpha/2,\alpha/2,0)$ and $q=(1-\beta,\beta/2,\beta/2,0)$, the composite has error law $r*p$.  For any two Pauli laws $a,b$,
\begin{equation}
 \frac12\|\Phi_a-\Phi_b\|_\diamond
 =\frac12\|a-b\|_1.
 \label{eq:pauli-isometry}
\end{equation}
The lower bound follows by applying the channels to one half of a maximally entangled state, which produces orthogonal Bell outcomes; the triangle inequality gives the reverse bound.  Consequently \cref{eq:deficiency} is a four-point total-variation linear program.

The source and target laws are invariant under exchanging $X$ and $Y$, so an optimizer can be averaged with its exchange and written
\begin{equation}
 r=(1-2x-z,x,x,z),
 \qquad x,z\ge0,\quad 2x+z\le1.
\end{equation}
Writing $s'=r*p$, direct convolution gives
\begin{align}
 s'_0&=1-\alpha-(2-3\alpha)x-(1-\alpha)z,\nonumber\\
 s'_1=s'_2&=\frac{\alpha}{2}+(1-2\alpha)x,\nonumber\\
 s'_3&=\alpha x+(1-\alpha)z.
 \label{eq:convolution-coordinates}
\end{align}
Assume first $d:=\beta-\alpha>0$ and $\alpha>0$.  Set
\begin{align}
 h&=s'_3, & \Delta&=q_0-s'_0,\nonumber\\
 c&=\frac{2-3\alpha}{\alpha}>1,
 &k&=\frac{2(1-\alpha)(1-2\alpha)}{\alpha}\geq0.
\end{align}
Eliminating $x$ from \cref{eq:convolution-coordinates} gives
\begin{equation}
 \Delta=-d+ch-kz.
 \label{eq:delta-relation}
\end{equation}
Normalization and the $X/Y$ symmetry reduce the objective to
\begin{equation}
 \TV(q,s')
 =\frac12\bigl(|\Delta|+|h-\Delta|+h\bigr).
 \label{eq:tv-three}
\end{equation}
If $0\le\Delta\le h$, then \cref{eq:delta-relation} implies $h\ge d/c$ and \cref{eq:tv-three} equals $h$.  If $\Delta<0$, then
\begin{equation}
 \TV(q,s')=h-\Delta
 =d-(c-1)h+kz>\frac dc+\frac{kz}{c}\ge\frac dc,
\end{equation}
where $\Delta<0$ implies $h<(d+kz)/c$.  Finally, if $\Delta>h$, then
\begin{equation}
 \TV(q,s')=\Delta>h>\frac{d+kz}{c-1}>\frac dc.
\end{equation}
Thus the optimum is at least $d/c$.  Equality is attained by $z=0$ and
$x=x_*=d/(2-3\alpha)$, for which $\Delta=0$ and
$h=\alpha d/(2-3\alpha)$.  The constraint $2x_*\le1$ follows from
$\beta<1/2$.  When $\alpha=0$, the same optimizer is simply
$r_*=q$ and converts the identity channel exactly, agreeing with the
continuous limit of the formula.

If $\beta<\alpha$, then $q_0=1-\beta>1-\alpha=\max_i p_i$.  For every
probability law $r$, the identity component of $r*p$ is a convex
combination of the components of $p$ and is at most $1-\alpha$.
Therefore
\begin{equation}
 \TV(q,r*p)\ge q_0-(r*p)_0\ge\alpha-\beta.
\end{equation}
Taking $r$ concentrated at the identity attains equality.  The case
$\alpha=\beta$ is immediate.
\end{proof}

\begin{corollary}[Operational discrimination meaning of the deficiency]
\label{cor:deficiency-discrimination}
For every CPTP converter $\Lambda$, the optimal equal-prior one-use success
probability for distinguishing the target channel $\Phi_\beta$ from the
converted source $\Lambda\circ\Phi_\alpha$ satisfies
\begin{equation}
 P_{\rm succ}(\Phi_\beta,\Lambda\Phi_\alpha)
 =\frac12+\frac14\|\Phi_\beta-\Lambda\Phi_\alpha\|_\diamond
 \geq\frac12+\frac12\delta_\diamond(\beta|\alpha).
 \label{eq:deficiency-operational}
\end{equation}
Equality is attained by the converters in
\cref{thm:exact-deficiency}.  Hence the unavoidable discrimination advantage
over random guessing is exactly
\begin{equation}
 \frac12\delta_\diamond(\beta|\alpha)=
 \begin{cases}
 (\alpha-\beta)/2,&\beta<\alpha,\\[2pt]
 \alpha(\beta-\alpha)/[2(2-3\alpha)],&\beta>\alpha.
 \end{cases}
 \label{eq:operational-advantage}
\end{equation}
\end{corollary}

\begin{proof}
The channel Helstrom theorem gives the first equality.  The definition of
$\delta_\diamond$ gives the lower bound, and the explicit optimal converters
from \cref{thm:exact-deficiency} attain the infimum.
\end{proof}

\begin{corollary}[A trace-distance-ordered Blackwell antichain]
\label{cor:antichain}
Fix a nonzero Wigner angle and the family $\{\Phi_w:0\le w\le1/2\}$.
The noiseless member $\Phi_0=\id$ post-processes to every $\Phi_w$.
Every two distinct members with positive weights are incomparable under
CPTP post-processing.  Nevertheless, if $0<u<v$, then
\cref{eq:all-pair-contraction} strictly orders \emph{every} pairwise
spin-state trace distance.
\end{corollary}

\begin{proof}
For $u,v>0$ and $u\ne v$, both directed quantities in
\cref{eq:exact-deficiency} are positive.  For $u=0$, choose the
converter $\Phi_v$.  The trace-distance statement is
\cref{thm:nogo}.
\end{proof}

\begin{figure}[t]
\centering
\includegraphics[width=\columnwidth]{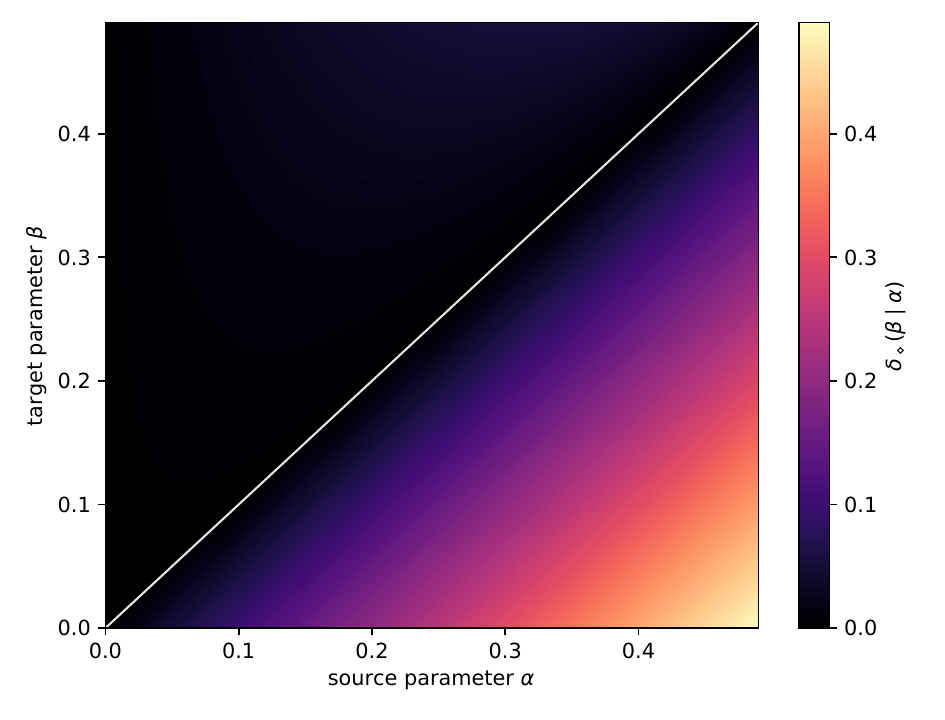}
\caption{Exact directed CPTP deficiency
$\delta_\diamond(\beta|\alpha)$ on the Pauli line.  The diagonal is the
only zero set in the positive-noise interior; the edge $\alpha=0$ is
also zero because every target is a post-processing of the identity.}
\label{fig:deficiency}
\end{figure}

For $\beta>\alpha$, the exact operational gap and the negative factor-Choi
eigenvalue are related by
\begin{equation}
 \delta_\diamond(\beta|\alpha)
 =
 \frac{2(1-\alpha)(1-2\alpha)}{2-3\alpha}
 \bigl[-\lambda_{\min}(\omega_T)\bigr].
 \label{eq:gap-vs-negativity}
\end{equation}
Thus the unphysical-factor witness and the best achievable physical
simulation error have the same leading scale but are not identical.

The failure is second order in the Wigner angle.  At small $\theta$,
\begin{equation}
 -\lambda_{\min}
 =\frac{u(v-u)}8\theta^4+O(\theta^6).
 \label{eq:small-angle}
\end{equation}
The exact deficiency has the same expansion,
\begin{equation}
 \delta_\diamond(sv|su)
 =\frac{u(v-u)}8\theta^4+O(\theta^6).
 \label{eq:small-angle-deficiency}
\end{equation}
Thus the nonrelativistic regime remains conceptually nondegradable but is statistically demanding.

\begin{figure}[t]
\centering
\includegraphics[width=\columnwidth]{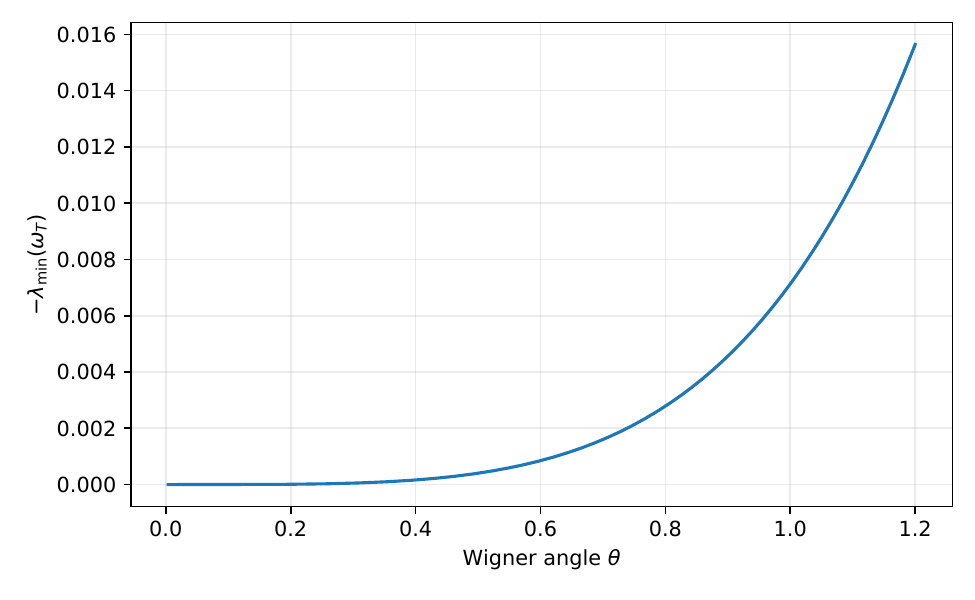}
\caption{Analytic nondegradability margin from \cref{eq:negative-eig}.  The effect is strict for every nonzero Wigner angle and every $0<u<v<1/2$.}
\label{fig:margin}
\end{figure}

\subsection{A finite-rapidity example}

Take $\theta=0.6$, $u=1/5$, $v=9/20$, and a fixed observer rapidity $\xi=1.6$.  Equation \eqref{eq:wigner-angle} gives particle rapidity
\begin{equation}
 \zeta=2\operatorname{artanh}
 \left(\frac{\tan(0.3)}{\tanh(0.8)}\right)
 \simeq1.00949.
\end{equation}
The corresponding speeds are $\tanh\xi\simeq0.92167$ and $\tanh\zeta\simeq0.76555$.  The Bloch matrices are
\begin{align}
 M_A&\simeq\diag(0.965067,0.965067,0.930134),\\
M_B&\simeq\diag(0.921401,0.921401,0.842802),
\end{align}
and $\lambda_{\min}\simeq-8.49661\times10^{-4}$.  The exact best
physical simulation error is
\begin{equation}
 \delta_\diamond(B|A)\simeq8.04866\times10^{-4},
\end{equation}
attained by \cref{eq:optimal-converter}; it is not merely bounded from
below by the unphysical factor.  Lower rapidities are obtained by
decreasing $\theta$, with the quartic margins in
\cref{eq:small-angle,eq:small-angle-deficiency}.

\section{A Pauli-Markov embedding obstruction}

The Blackwell antichain is a comparison statement between different boost
parameters.  A complementary question is whether one nontrivial channel on
the line can itself be obtained at positive time from a time-homogeneous
Pauli Lindblad semigroup.

\begin{theorem}[No nontrivial Pauli-semigroup embedding]
\label{thm:pauli-semigroup-obstruction}
Let $0<\alpha<1/2$ and let $\Phi_\alpha$ be the channel
\eqref{eq:pauli-line}.  There are no rates
$\gamma_x,\gamma_y,\gamma_z\geq0$ and no $t>0$ such that
\begin{equation}
 \Phi_\alpha=e^{t\mathcal L},\qquad
 \mathcal L(\rho)=\sum_{j=x,y,z}
 \gamma_j(\sigma_j\rho\sigma_j-\rho).
 \label{eq:pauli-lindblad-generator}
\end{equation}
Equivalently, every nonidentity channel on this fixed-boost Pauli line violates the exact
Markov-embedding inequality
\begin{equation}
 \lambda_z\geq\lambda_x\lambda_y
 \label{eq:pauli-embedding-inequality}
\end{equation}
for positive Pauli eigenvalues.  Hence the one-shot relativistic noise is not
the positive-time point of any time-homogeneous purely Pauli dissipative
semigroup, even though it is a perfectly valid random-unitary channel.
\end{theorem}

\begin{proof}
A semigroup generated by \eqref{eq:pauli-lindblad-generator} has Bloch
eigenvalues
\begin{align*}
 \lambda_x&=e^{-2t(\gamma_y+\gamma_z)},
 &\lambda_y&=e^{-2t(\gamma_x+\gamma_z)},\\
 \lambda_z&=e^{-2t(\gamma_x+\gamma_y)}.
\end{align*}
Therefore $\lambda_z/(\lambda_x\lambda_y)=e^{4t\gamma_z}\geq1$, proving
\eqref{eq:pauli-embedding-inequality}.  For $\Phi_\alpha$,
$\lambda_x=\lambda_y=1-\alpha$ and $\lambda_z=1-2\alpha$, whereas
\[
 1-2\alpha<(1-\alpha)^2
\]
for every $\alpha>0$.  This contradicts
\eqref{eq:pauli-embedding-inequality}.  Equivalently, solving for the rates
would give
\[
 \gamma_z=\frac1{4t}
 \log\frac{1-2\alpha}{(1-\alpha)^2}<0.
\]
\end{proof}

\begin{remark}
The theorem excludes the natural Pauli-diagonal Markov embedding.  It is not
stated as a classification of embeddings into arbitrary nonnormal qubit
Lindbladians with Hamiltonian rotations.  The exact directed deficiency of
Theorem~\ref{thm:exact-deficiency} remains the stronger operational
comparison result used in the paper.
\end{remark}

\subsection{Robust exclusion radius}
\label{sec:robust-markov-exclusion}

The semigroup obstruction is stable under tomography and model error.

\begin{theorem}[Certified exclusion neighborhood for Pauli semigroups]
\label{thm:robust-pauli-semigroup-exclusion}
Let
\[
 \lambda(\alpha)=(1-\alpha,1-\alpha,1-2\alpha),
 \qquad 0<\alpha<\frac12,
\]
be the Bloch eigenvalue triple of a nontrivial fixed-boost channel.  Let
$\mu=(\mu_x,\mu_y,\mu_z)$ be any Pauli channel satisfying
$\|\mu-\lambda(\alpha)\|_\infty\leq\delta$.  If
\[
 \delta<\frac{\alpha^2}{3},
\]
then $\mu$ cannot be a positive-time element of a time-homogeneous Pauli-diagonal
Lindblad semigroup.
\end{theorem}

\begin{proof}
Every Pauli-semigroup element obeys the multiplicative inequality
$\mu_z\geq\mu_x\mu_y$.  At the fixed-boost point,
\[
 \lambda_x\lambda_y-\lambda_z=(1-\alpha)^2-(1-2\alpha)=\alpha^2.
\]
Since all Pauli eigenvalues lie in $[-1,1]$,
\[
 |\mu_x\mu_y-\lambda_x\lambda_y|
 \leq |\mu_x-\lambda_x||\mu_y|
      +|\lambda_x||\mu_y-\lambda_y|
 \leq2\delta.
\]
Together with $|\mu_z-\lambda_z|\leq\delta$ this gives
$\mu_x\mu_y-\mu_z\geq\alpha^2-3\delta>0$, contradicting the semigroup
inequality.
\end{proof}

\begin{corollary}[Tomographic certificate]
If simultaneous confidence intervals for the three Bloch eigenvalues have
half-width smaller than $\alpha^2/3$ around the fitted fixed-boost point, the
whole confidence box is excluded from the Pauli-Markovian semigroup region.
\end{corollary}

\begin{proof}
Every point of the confidence box lies within $\ell_\infty$ distance $\delta<\alpha^2/3$ of the fitted fixed-boost triple.  Applying \Cref{thm:robust-pauli-semigroup-exclusion} pointwise excludes the entire box.
\end{proof}

\section{Normalizable wave packets and open stability}

Sharp momentum branches are a convenient exact model but are not normalizable wave packets.  The no-go is not tied to this idealization.

\begin{proposition}[Wave-packet persistence]
\label{prop:wavepacket}
For every parameter triple $(\theta,u,v)$ in \cref{thm:nogo}, there exist pure, normalizable momentum states $|\psi_u\rangle$ and $|\psi_v\rangle$ with five disjoint compact momentum-space components such that the resulting source and target channels retain the strict all-state contraction \cref{eq:all-pair-contraction} and have a factor with a strictly negative Choi eigenvalue.  The same remains true under sufficiently small changes of packet shapes, branch weights, and momentum directions.
\end{proposition}

\begin{proof}
Let $p_0,p_{\pm x},p_{\pm y}$ be the five ideal momenta, with $p_0$ chosen parallel to the boost so that its Wigner rotation is the identity.  Choose pairwise disjoint relatively compact neighborhoods $N_j$ and normalized functions $\varphi_j\in C_c^\infty(N_j)$.  For a weight $w\in(0,1/2)$ define
\begin{equation}
 \psi_w
 =\sqrt{1-2w}\,\varphi_0
 +\sqrt{\frac w2}\,
  (\varphi_{+x}+\varphi_{-x}+\varphi_{+y}+\varphi_{-y}).
 \label{eq:five-packet}
\end{equation}
Disjoint support gives $\|\psi_w\|_2=1$ and
\begin{equation}
 |\psi_w(p)|^2
 =(1-2w)|\varphi_0(p)|^2
 +\frac w2\sum_{\tau\in\{\pm x,\pm y\}}|\varphi_\tau(p)|^2.
\end{equation}
Thus the reduced spin channel has precisely the intended packet weights; no diagonal generalized momentum state is required.

For a massive particle and a fixed finite boost,
$p\mapsto R_W(\Lambda,p)$ is continuous on the positive-energy mass
shell.  The averaged Bloch matrices converge uniformly to $M(w)$ as
the diameters of all $N_j$ shrink.  Define the positive pairwise
contraction margin
\begin{equation}
 g=\lambda_{\min}\!\left(M(u)^TM(u)-M(v)^TM(v)\right)>0.
\end{equation}
The Choi margin $-\lambda_{\min}(\omega_T)$ in
\cref{eq:negative-eig} is also positive.  Matrix multiplication,
inversion on the open set of invertible matrices, and minimum
eigenvalues are continuous.  For sufficiently narrow packets, the
perturbed contraction matrix remains positive definite with margin at
least $g/2$, and the factor Choi state retains an eigenvalue below
$\lambda_{\min}(\omega_T)/2<0$.  The same strict inequalities persist
under all sufficiently small parameter and shape perturbations.
\end{proof}

Asymmetric packets need not preserve the diagonal form.  The proposition consequently yields nondiagonal, genuinely noncommuting examples within the same fixed-boost kinematics.

\begin{corollary}[Quantitative packet-width criterion]
\label{cor:packet-width}
Assume the Wigner rotation is $L$-Lipschitz in operator norm on the
union of the five packet neighborhoods and every neighborhood has
diameter at most $h$.  Then the packet-averaged Bloch matrices obey
\begin{equation}
 \|\widetilde M_w-M(w)\|_2\le\varepsilon:=Lh.
 \label{eq:packet-error}
\end{equation}
Let
\begin{equation}
 g=\lambda_{\min}\!\left(M(u)^TM(u)-M(v)^TM(v)\right).
\end{equation}
The strict all-state contraction persists whenever
\begin{equation}
 g>4\varepsilon+2\varepsilon^2.
 \label{eq:packet-contraction-condition}
\end{equation}
If, in addition, the factor-Choi perturbation bound obtained from
\cref{eq:factor-error,eq:choi-error} is smaller than
$-\lambda_{\min}(\omega_T)$, nondegradability is certified for the
finite-width packets.
\end{corollary}

\begin{proof}
The first bound follows by integrating the pointwise Lipschitz
estimate.  For any contraction $M$ and perturbation $E$ with
$\|E\|_2\le\varepsilon$,
\begin{equation}
 \|(M+E)^T(M+E)-M^TM\|_2
 \le2\varepsilon+\varepsilon^2.
\end{equation}
Applying this to source and target shows that the positive matrix
$M(u)^TM(u)-M(v)^TM(v)$ loses at most
$4\varepsilon+2\varepsilon^2$ in its smallest eigenvalue.  The Choi
statement is Weyl eigenvalue stability combined with
\cref{eq:factor-error,eq:choi-error}.
\end{proof}

\begin{theorem}[Closed perturbation radius for the fixed-boost no-go]
\label{thm:closed-perturbation-radius}
Put $\alpha=su$, $\beta=sv$, $0<\alpha<\beta<1/2$, and define
\begin{align}
 r&=1-2\alpha,\nonumber\\
 \eta&=\frac{\alpha(\beta-\alpha)}
 {2(1-\alpha)(1-2\alpha)},\label{eq:closed-radius-eta}\\
 g&=\min\bigl\{(\beta-\alpha)(2-\alpha-\beta),
 4(\beta-\alpha)(1-\alpha-\beta)\bigr\}.\label{eq:closed-radius-g}
\end{align}
Here $\eta=-\lambda_{\min}(\omega_T)$ and $g$ is the exact smallest
eigenvalue of
$M_\alpha^TM_\alpha-M_\beta^TM_\beta$.  Suppose perturbed source and target are unital qubit channels whose
Bloch matrices obey
\[
 \|\widetilde M_\alpha-M_\alpha\|_2\leq\varepsilon,
 \qquad
 \|\widetilde M_\beta-M_\beta\|_2\leq\varepsilon.
\]
Define
\begin{align}
 \varepsilon_{\rm CP}
 &=\frac{2\eta r^2}
 {\sqrt3(r+1)+2\eta r},\label{eq:closed-radius-cp}\\
 \varepsilon_{\rm con}
 &=\sqrt{1+\frac g2}-1.\label{eq:closed-radius-con}
\end{align}
If
\begin{equation}
 \varepsilon<\min\{\varepsilon_{\rm CP},\varepsilon_{\rm con}\},
 \label{eq:closed-radius-condition}
\end{equation}
then the perturbed target strictly contracts every nonzero Bloch difference
more than the perturbed source, while the unique factor
$\widetilde T=\widetilde M_\beta\widetilde M_\alpha^{-1}$ still has a
negative normalized Choi eigenvalue.  Thus
\eqref{eq:closed-radius-condition} is an explicit open ball of physical
counterexamples, with no unspecified continuity constant.
\end{theorem}

\begin{proof}
The least singular value of $M_\alpha$ is $r$.  Whenever
$\varepsilon<r$, the resolvent identity and $\|M_\beta\|_2\leq1$ give
\begin{equation}
 \|\widetilde T-T\|_2
 \leq\frac{\varepsilon(r+1)}{r(r-\varepsilon)}.
 \label{eq:closed-radius-factor-bound}
\end{equation}
The normalized Choi perturbation is at most $\sqrt3/2$ times this quantity,
by the Pauli-basis estimate \eqref{eq:choi-error}.  Requiring it to be less
than $\eta$ is algebraically equivalent to
$\varepsilon<\varepsilon_{\rm CP}$; this inequality itself implies
$\varepsilon<r$.

For contraction, each Gram matrix changes by at most
$2\varepsilon+\varepsilon^2$.  Hence the source-minus-target Gram difference
retains a positive least eigenvalue whenever
$g>4\varepsilon+2\varepsilon^2$, which is exactly
$\varepsilon<\varepsilon_{\rm con}$.  The two strict conditions together
prove the theorem.
\end{proof}

\section{Robust tomography certificate}

Suppose tomography gives $\widehat M_A,\widehat M_B$ with
\begin{equation}
 \|\widehat M_A-M_A\|_2\le\varepsilon_A,
 \qquad
 \|\widehat M_B-M_B\|_2\le\varepsilon_B.
\end{equation}
Let $\widehat s=\smin(\widehat M_A)$ and assume $\widehat s>\varepsilon_A$.  Then $M_A$ and $\widehat M_A$ are invertible.  For $T=M_BM_A^{-1}$ and $\widehat T=\widehat M_B\widehat M_A^{-1}$, trace-distance contraction of the physical channel $B$ gives $\|M_B\|_2\le1$, and the resolvent identity yields the fully data-dependent bound
\begin{equation}
 \|\widehat T-T\|_2
 \le
 \Delta_T:=
 \frac{\varepsilon_B}{\widehat s}
 +\frac{\varepsilon_A}{\widehat s(\widehat s-\varepsilon_A)}.
 \label{eq:factor-error}
\end{equation}
Indeed,
\begin{equation}
 \widehat T-T=(\widehat M_B-M_B)\widehat M_A^{-1}
 +M_B\widehat M_A^{-1}(M_A-\widehat M_A)M_A^{-1},
\end{equation}
and $\|M_A^{-1}\|_2\le(\widehat s-\varepsilon_A)^{-1}$.

The normalized Choi state of a unital qubit map is
\begin{equation}
 \omega_T=\frac14\left(I\otimes I+
 \sum_{i,j=1}^3T_{ij}\sigma_i\otimes\sigma_j^T\right).
\end{equation}
Using $\|X\|_F\le\sqrt3\|X\|_2$ and the Hilbert--Schmidt orthogonality of Pauli products gives the conservative bound
\begin{equation}
 \|\omega_{\widehat T}-\omega_T\|_\infty
 \le\frac{\sqrt3}{2}\|\widehat T-T\|_2.
 \label{eq:choi-error}
\end{equation}

\begin{corollary}[Finite-error rejection rule]
\label{cor:rejection}
If
\begin{equation}
 \lambda_{\min}(\omega_{\widehat T})
 +\frac{\sqrt3}{2}\Delta_T<0,
 \label{eq:reject}
\end{equation}
where $\Delta_T$ is the certified bound in \cref{eq:factor-error}, then no CPTP post-processing exists.  Equivalently, one may propagate a confidence region for $(M_A,M_B)$ and maximize the minimum factor-Choi eigenvalue over that region.
\end{corollary}

\begin{proof}
Weyl's eigenvalue inequality and \eqref{eq:choi-error} give
\begin{equation*}
 \lambda_{\min}(\omega_T)\leq\lambda_{\min}(\omega_{\widehat T})
 +\frac{\sqrt3}{2}\Delta_T.
\end{equation*}
Under \eqref{eq:reject}, the true factor Choi state has a negative eigenvalue, so the unique factor cannot be completely positive.
\end{proof}

\begin{corollary}[Finite-shot confidence certificate]
\label{cor:finite-shot}
For each channel, prepare the three positive Pauli eigenstates and
estimate the three output Pauli expectations with $N$ independent
binary shots per matrix entry.  At confidence at least $1-\gamma$, one
may use
\begin{equation}
 \varepsilon_A=\varepsilon_B
 =3\sqrt{\frac{2}{N}\log\frac{36}{\gamma}}
 \label{eq:shot-error}
\end{equation}
in \cref{eq:factor-error,eq:reject}.  In particular, a prescribed
spectral-norm radius $\varepsilon$ is certified whenever
\begin{equation}
 N\ge\frac{18}{\varepsilon^2}\log\frac{36}{\gamma}
 \label{eq:shot-complexity}
\end{equation}
shots are taken for each setting (nine settings per channel).
\end{corollary}

\begin{proof}
For a unital channel, the expectation of $\sigma_i$ on the output
generated by the positive $\sigma_j$ eigenstate equals $M_{ij}$.  Each
sample is in $\{-1,+1\}$.  Hoeffding's inequality and a union bound over
the $18$ entries give
\begin{equation}
 \max_{C\in\{A,B\}}\max_{i,j}
 |\widehat M_{C,ij}-M_{C,ij}|
 \le \sqrt{\frac{2}{N}\log\frac{36}{\gamma}}
\end{equation}
with probability at least $1-\gamma$.  A $3\times3$ matrix with every
entry bounded by $a$ has spectral norm at most its Frobenius norm,
which is at most $3a$.  Substitution proves both claims.
\end{proof}

Only three linearly independent spin preparations and three Pauli measurement axes are needed for unital process tomography.  The relevant conclusion is not a comparison of separate decay fits but a one-sided confidence statement about Choi positivity.

\section{Physical meaning and limitations}

The target channel in \cref{thm:nogo} is noisier according to every
principal polarization contrast, every singular value, and every
binary spin-state trace distance.  Yet it is not a post-processing of
the source.  The missing condition is joint consistency on entangled
inputs.  In the factor, the longitudinal contraction is too strong
relative to the two transverse contractions to fit inside the CPTP
tetrahedron.  The exact deficiency
\cref{eq:exact-deficiency} shows that this is not an artifact of
choosing the algebraic inverse: no alternative physical converter can
remove the residual gap.

The model keeps the observer boost fixed and modifies only a calibrated momentum distribution.  It therefore represents a more restrictive relativistic setting than a classical mixture of different boosts.  Nevertheless it still assumes a product spin--momentum preparation.  Outside that compatibility domain, reduced spin dynamics can depend on the assignment of correlated initial states and should be treated at the joint level.

The exact cone assumes equal transverse momentum magnitude and inversion symmetry.  General wave packets are covered locally by \cref{prop:wavepacket}, not globally classified.  Higher spin replaces the qubit tetrahedron with a larger Choi cone and requires a separate analysis.

\section{Conclusion}

A single fixed Lorentz boost, combined with different calibrated momentum preparations, can generate reduced spin channels whose observable contractions are totally ordered while their information structures are not.  The inversion-symmetric transverse shell has an
exact matrix cone.  Inside its analytic Pauli line, the identity
dominates every member, but all positive-noise members form a
Blackwell antichain.  We obtain the exact directed diamond deficiency
and its optimal physical converter, not only a non-CP inverse witness.
The effect survives pure normalizable wave packets and nondiagonal perturbations, with explicit stability and tomography radii.  Consequently, even the simultaneous decay of all
binary spin distinguishabilities is not a complete monotone for
fixed-boost relativistic spin noise.

\appendix

\section{Trace-class wave packets and the random-unitary reduction}
\label{app:wavepacket-reduction}

We record the momentum trace carefully because generalized momentum
eigenkets can obscure the state-space domain.  Let
$|\chi\rangle=\sum_\sigma c_\sigma|\sigma\rangle$ and
\begin{equation}
 |\Psi\rangle
 =\int_{\mathcal H_m^+}d\nu_m(p)\,
 \psi(p)|p\rangle\otimes|\chi\rangle.
\end{equation}
Using covariantly normalized momentum kets, the transformed vector is
\begin{equation}
 U(\Lambda)|\Psi\rangle
 =\int d\nu_m(p)\,\psi(p)|\Lambda p\rangle
 \otimes D(W(\Lambda,p))|\chi\rangle.
\end{equation}
For a bounded spin observable $A$,
\begin{multline}
 \langle U(\Lambda)\Psi|I\otimes A|U(\Lambda)\Psi\rangle\\
 =\int d\nu_m(p)\,|\psi(p)|^2
 \langle\chi|D(W(\Lambda,p))^\dagger
 A D(W(\Lambda,p))|\chi\rangle.
 \label{eq:weak-partial-trace}
\end{multline}
The cross terms vanish because the Lorentz transformation is
one-to-one on the mass shell and the invariant measure absorbs the
Jacobian.  Since \cref{eq:weak-partial-trace} holds for every bounded
$A$, it identifies the reduced density operator and proves
\cref{eq:reduced-channel} for pure packets.  A positive trace-class
$\rho_P$ has a spectral decomposition
$\rho_P=\sum_\ell\lambda_\ell|\psi_\ell\rangle\langle\psi_\ell|$;
linearity and monotone convergence give the mixed-state formula with
\begin{equation}
 \mu(dp)=\sum_\ell\lambda_\ell|\psi_\ell(p)|^2d\nu_m(p).
\end{equation}
In particular, the construction \cref{eq:five-packet} is an ordinary
rank-one state even though its reduced channel is a convex mixture of
five narrow momentum components.

\section{Derivation of the Wigner-angle parametrization}

For two perpendicular canonical boosts with rapidities $\xi$ and $\zeta$, the standard half-angle formula is \cref{eq:wigner-angle}.  Given a fixed observer rapidity $\xi$, every angle satisfying
\begin{equation}
 \tan(\theta/2)<\tanh(\xi/2)
\end{equation}
is realized by the finite particle rapidity
\begin{equation}
 \zeta=2\operatorname{artanh}
 \left(\frac{\tan(\theta/2)}{\tanh(\xi/2)}\right).
\end{equation}
Changing the transverse momentum direction rotates the Wigner axis in the transverse plane; reversing momentum reverses that axis.

\section{Reproducibility}

The supplied script \texttt{verify\_fixed\_boost\_wigner.py} checks the
cone parametrization, the exact Choi spectrum, the finite-rapidity
example, and $1000$ random parameter triples.  Independently of the
analytic proof, it solves the four-point post-processing linear program
for $1000$ random source--target pairs and compares every optimum with
\cref{eq:exact-deficiency}.  It also regenerates all three figures.
These calculations test implementation and normalization; the theorem
signs and the linear-program optimum are analytic.

\bibliography{references}
\end{document}